\documentclass[article,preprint,showpacs,amsfonts]{revtex4}
\usepackage{indentfirst}
\usepackage{amsmath,amssymb}
\usepackage{amsthm}
\usepackage{graphicx}
\usepackage{multirow}
\usepackage{CJK}

\usepackage[utf8]{inputenc}

\usepackage{amsfonts}
\usepackage{bm}

\newcommand{{\Cd}}{{\mathbb{C}^d}}

\newcommand{\bra}[1]{\mbox{$\langle #1 |$}}

\newcommand{\ket}[1]{\mbox{$| #1 \rangle$}}

\def\<{\langle}
\def\>{\rangle}
\newtheorem{Theorem}{Theorem}
\newtheorem{Lemma}{Lemma}

\newtheorem{Corollary}{Corollary}

\newtheorem{Example}{Example}

\newcommand{\beq}{\begin{equation}}
\newcommand{\eeq}{\end{equation}}
\newcommand{\bear}{\begin{eqnarray}}
\newcommand{\ear}{\end{eqnarray}}
\newcommand{\bdm}{\begin{displaymath}}
\newcommand{\edm}{\end{displaymath}}

\begin{document}
	\begin{CJK}{UTF8}{song}	
\title{ The stabilizer for $n$-qubit symmetric states }
	\author{Xian Shi}
	\affiliation{Institute of Mathematics, Academy of Mathematics and Systems Science, Chinese Academy of Sciences, Beijing 100190, China
		\\
		University of Chinese Academy of Sciences, Beijing 100049, China
		\\
		UTS-AMSS Joint Research Laboratory for Quantum Computation and Quantum Information Processing,
		Academy of Mathematics and Systems Science, Chinese Academy of Sciences, Beijing 100190, China}
	\email{shixian01@gmail.com}
	\begin{abstract}
	\indent The stabilizer group for an $n$-qubit state $\ket{\phi}$ is the set of all invertible local operators (ILO) $g=g_1\otimes g_2\otimes \cdots\otimes g_n,$ $ g_i\in \mathcal{GL}(2,\mathbb{C})$ such that $\ket{\phi}=g\ket{\phi}.$ Recently, [Gour $et$ $al.$ 2017 J. Math. Phys. (N.Y.) \textbf{58} 092204] presented that almost all $n$-qubit state $\ket{\psi}$ own a trivial stabilizer group when $n\ge 5.$ In this article, we consider the case when the stabilizer group of an $n$-qubit symmetric pure state $\ket{\psi}$ is trivial. First we show that the stabilizer group for an n-qubit symmetric pure state $\ket{\phi}$ is nontrivial when $n\le 4$. Then we present a class of $n$-qubit symmetric states $\ket{\phi}$ with trivial stabilizer group when $n\ge 5$. At last, we prove that when $m\ge 5$, almost all $n$-qubit symmetric pure states own a trivial stabilizer group, due to the results in [Gour $et$ $al.$ 2017 J. Math. Phys. (N.Y.) \textbf{58} 092204], we have that almost all $n$-qubit symmetric pure states are isolated.
	\end{abstract}
\maketitle
\end{CJK}
\section{Introduction}
\indent Quantum entanglement \cite{HHH} is a valuable resource for a variety of tasks that cannot be finished by classical resource. Among the most popular tasks are quantum teleportation \cite{BB} and quantum superdense coding \cite{BS}. Due to the importance of quantum entanglement, the classification of quantum entanglement states is a big issue for the quantum information theory.\\
\indent Entanglement theory is a resource theory with its free transformation is local operations and classical communication (LOCC). As LOCC is hard to deal with mathematically, and with the number of the parties of the quantum systems grows, the classification of all entanglement states under the LOCC restriction becomes very hard. A conventional way is to consider other operations, such as stochastic LOCC (SLOCC), local unitary operations (LU), separable operations (SEP).\\
\indent Two n-partite states $\ket{\phi}$ and $\ket{\psi}$ are SLOCC equivalent \cite{WGJ} if and only if there exists n invertible local operations (ILO) $A_i,i=1,2,\cdots,n$ such that
\begin{align}
\ket{\phi}=A_1\otimes A_2\otimes \cdots\otimes A_n\ket{\psi}.
\end{align}
Then the classification for multi-qubit pure states under SLOCC attracts much attention \cite{MBDM, CC,L,ZZH, LL, GW}. However, there are uncountable number of SLOCC inequivalent classes in $n$-qubit systems when $n\ge 4,$ so it is a formidable task to classify multipartite states under SLOCC. \\
\indent SEP is simple to describe mathematically and contains LOCC strictly, as there exists pure state transformations belonging to SEP, but cannot be achieved by LOCC. The authors in \cite{GRW} presented that the existence of transformations under separable operations between two pure states  depends largely on the stabilizer of the state. Recently, Gour $et$ $al.$ showed that almost all of the stabilizer group for 5 or more qubits pure states contains only the identity \cite{GKW}. And the authors in \cite{SW} generalized this result to n-qudit systems when $n>3,d>2.$\\
\indent Symmetric states belong to the space that is spanned by the pure states invariant under particle exchange, and there are some results done on the classifications under SLOCC limited to symmetric states \cite{ TSPM,PSM,PR,PJM,BBM}. The authors in \cite{PSM} proved if $\ket{\psi}$ and $\ket{\phi}$ are $n$-qubit symmetric pure states, and there exists n invertible operations $A_i,i=1,2,\cdots,n$ such that $\ket{\psi}=A_1\otimes A_2\otimes \cdots \otimes A_n\ket{\phi},$ then there exists an invertible matrix A such that
\begin{align}
\ket{\psi}=A^{\otimes n}\ket{\phi}.
\end{align}
Moreover, P. Migdal $et$ $al.$ \cite{PJM} generalized the results from qubit systems to qudit systems.\\
\indent In the article, we consider the problem on the stabilizer groups for $n$-qubit symmetric states. This article is organized as follows. In section \uppercase\expandafter{\romannumeral2},
we present preliminary knowledge on $n$-qubit symmetric pure states. In section \uppercase\expandafter{\romannumeral3}, we present our main results. First, we present the stabilizer group for an $n$-qubit symmetric state is nontrivial when $n\le 4$, then we present a class of $n$-qubit symmetric states whose stabilizer group is trivial when $n\ge 5$, at last, when $m= 2,$ the pure state $\ket{\psi}$ owns a nontrivial stabilizer group, when $m=4,$ there exists only one case when $\ket{\psi}$ owns a nontrivial stabilizer group, when $m\ge 5,$ the stabilizer group of almost all $n$-qubit symmetric pure state is trivial. In section \uppercase\expandafter{\romannumeral4}, we will end with a summary.
\section{Preliminary Knowledge}
\indent In this section, we will first recall the definition of symmetric states, and then we present the Majorana representation for an $n$-qubit symmetric pure state.\\
\indent A pure state $\ket{\psi}\in \mathcal{H}_2$ can be represented by a point on
a Bloch sphere geometrically as $\ket{\psi} = \cos \frac{\theta}{2}
\ket{0} + e^{i\phi} \sin\frac{\theta}{
2}
\ket{1},$ here two parameters $\theta\in[0, \pi], \phi\in [0, 2\pi).$ 
	
\begin{figure}
\centering
\includegraphics[width=60mm]{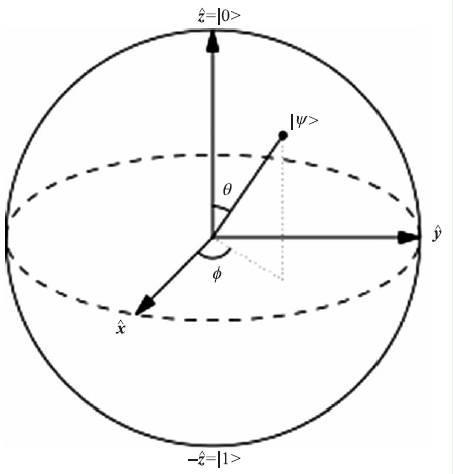}\\
\caption{Bloch Sphere}\label{Figure 1}
\end{figure}
We call an $n$-partite pure state $\ket{\psi}$ symmetric state if it is invariant under permuting the particles. That is, for any permutation
operator $P_{\pi}, P_{\pi}\ket{\psi} = \ket{\psi}.$
Generally, there are two main characterizations for an $n$-qubit symmetric pure state $\ket{\psi}$, Majorana representation \cite{M} and
Dicke representation . The Majorana representation for an n-qubit symmetric pure state is that there exists single particles $\ket{\phi_i},i=1,2,\cdots,n,$ such that
\begin{align}
\ket{\psi}=\frac{e^{i\theta}}{\sqrt{K}}\sum_{\sigma\in \textit{perm}} \ket{\phi_{\sigma(1)}} \ket{\phi_{\sigma(2)}}\cdots\ket{\phi_{\sigma(n)}},
\end{align}
where the sum runs over all distinct permutations $\sigma$, $K$ is a normalization prefactor and the $\ket{\phi_i}$ are single qubit states $\ket{\phi_i}=\cos \frac{\theta_i}{2}
\ket{0} + e^{i\phi_i} \sin\frac{\theta_i}{2}
\ket{1},i=1,2,\cdots n,$ we would denote $\ket{\phi_i}$ as $\ket{\phi_i}=a_i\ket{0}+b_i\ket{1}$ below.
 An n-qubit symmetric pure state $\ket{\psi}$ can also be characterized
as the sum of the Dicke states $\ket{D(n, k)}$, that is,
\begin{align}
\ket{\psi}=\sum x_k\ket{D(n,k)}.
\end{align} Here the Dicke states $\ket{D(n,k)}$ are defined as
\begin{align}
\ket{D(n,k)}=\frac{1}{\sqrt{C_n^k}}\sum \ket{\underbrace{0\cdot\cdot\cdot 0}_{n-k}\underbrace{1\cdot \cdot\cdot 1}_k},
\end{align}where the sum runs over all the permutations of the qubits. Up to a global phase factor, the parameters $a_i/b_i$ in a pure state $\ket{\phi_i}= a_i
\ket{0} + b_i
\ket{1}$
are
the roots of the polynomial $P(z) = \sum(-1)^k x_k\sqrt{C_n^k} z^k,$ here $C_n^k$ denotes the binomial coefficient of $n$ and $k.$  \\
\indent Next we introduce an isometric linear map
\begin{align}
\mathcal{M}: \mathcal{L}(\mathcal{H}_1,\mathcal{H}_2)\longrightarrow \mathcal{H}_1\otimes \mathcal{H}_2\\
\ket{b}\bra{a}\longrightarrow \ket{b}\ket{a}
\end{align}
here we denote that $\mathcal{L}(\mathcal{H}_1,\mathcal{H}_2)$ is the set of all linear maps from the Hilbert space $\mathcal{H}_1$ to the Hilbert space $\mathcal{H}_2.$ This map is useful for considering the stabilizer group for a $2$-qubit symmetric state. Now we introduce some properties of this map:\\
$(1).$ Assume $A\in \mathcal{L}(\mathcal{H}_1),B\in \mathcal{L}(\mathcal{H}_1),$ we have that
\begin{align}
A\otimes B\mathcal{M}(X)=\mathcal{M}(AXB^{T}).
\end{align}
$(2).$ Assume $\ket{\psi}\in\mathcal{H}_1,\ket{\phi}\in\mathcal{H}_2,$ then
\begin{align}
\mathcal{M}(\ket{\phi}\bra{\psi})=\ket{\phi}\ket{\bar{\psi}}.
\end{align}
\indent Then we recall the definition and some important properties of $M\ddot{o}bius$ transformation, which is useful for the last part of this article. $M\ddot{o}bius$ transformation is defined on the extended complex plane onto itself \cite{TN}, it can be represented as
\begin{align}
f(z)=\frac{az+b}{cz+d}, 
\end{align}            
with $a,b,c ,d\in \mathbb{C},ad-bc\ne 0.$ From the above equality, we see that when $c\ne 0,$ this function $f:\mathbb{C}-\{-d/c\}\rightarrow \mathbb{C}-\{a/c\}, f(-d/c)=\infty,f(\infty)=a/c,$ when $c=0,$ this function $f:\mathbb{C}\rightarrow\mathbb{C},f(\infty)=\infty.$ The $M\ddot{o}bius$ transformation owns the following properties:\\
$(1).$ $M\ddot{o}bius$ transformation map circles to circles.\\
$(2).$ $M\ddot{o}bius$ transformation are conformal.\\
$(3).$ If two points are symmetric with respect to a circle, then their images under a $M\ddot{o}bius$ transformation are symmetric with respect to the image circle. This is called the "Symmetry Principle."\\
$(4).$ With the exception of the identity mapping, a $M\ddot{o}bius$ transformation has at most two fixed points.\\
$(5).$ There exists a unique $M\ddot{o}bius$ transformation sending any three points to any other three points.\\
$(6).$ The unique $M\ddot{o}bius$ transformation $z\rightarrow M(z)$ sending three points $q,r,s$ to any other three points $q^{'},r^{'},s^{'}$ is given by 
\begin{align}
\frac{(M(z)-q^{'})(r^{'}-s^{'})}{(M(z)-s^{'})(r^{'}-q^{'})}=\frac{(z-q)(r-s)}{(z-s)(r-q)}.\nonumber
\end{align} 
$(7).$ The $M\ddot{o}bius$ transformation forms a group, $M\ddot{o}bius$ transformation is isotropic to the projective linear group $PSL(2,\mathbb{C})\cong SL(2,\mathbb{C})/\{I,-I\}.$\\
\indent 
As we know, the stereographic projection is a mapping that projects a sphere onto a plane. This projection is defined on the whole sphere except a point, and this map is smooth and bijective. It is conformal, $i.e.$ it preserves the angels at where curves meet. By transforming the majorana points of a pure state $\ket{\psi}$ to an extended complex plane, we may get the following proposition \cite{MW}.   
\begin{Lemma}
	Assume $\ket{\psi_1},\ket{\psi_2}$ are two pure symmetric states, if $\ket{\psi_{1}}$ and $\ket{\psi_{2}}$ are SLOCC-equivalent iff there exists a $M\ddot{o}bius$ transformation $(10)$ between their Majorana points. 
\end{Lemma} 
\indent At last, we recall two parameters defined in \cite{TSPM}, diversity degree and degeneracy configuration of an $n$-qubit symmetric pure state. Both two parameters can be used to identify the SLOCC entanglement classes of all $n$-qubit symmetric pure state. Assume $\ket{\psi}$ is an $n$-qubit symmetric pure state, $\ket{\psi}=\frac{e^{i\theta}}{\sqrt{K}}\sum_{\sigma\in \textit{perm}}\ket{\phi_{\sigma(1)}}\ket{\phi_{\sigma(2)}}\cdots\ket{\phi_{\sigma(n)}},$ up to a global phase factor, two states $\ket{\phi_i}$ and $\ket{\phi_j}$ are identical if and only if $a_ib_j-a_jb_i=0,$ and we define their number the degeneracy number. Then we define the degeneracy configuration $\{n_i\}$ of a symmetric state $\ket{\psi}$ as the list of its degeneracy numbers $n_i$ ordered in decreasing order. We denote the number of the elements in the set $\{n_i\}$ as the diversity degree $m$ of the symmetric state, it stands for the number of distinct $\ket{\phi_i}$ in the Eq.$(3).$ For example, a $3$-qubit GHZ state $\ket{GHZ}=\frac{\ket{000}+\ket{111}}{\sqrt{2}},$ we have $\ket{\phi_1}=\frac{\ket{0}+w\ket{1}}{2},\ket{\phi_2}=\frac{\ket{0}+w^2\ket{1}}{2},\ket{\phi_3}=\frac{\ket{0}+\ket{1}}{2},$ $w$ is roots of $P(z)=1-z^3,$ the degeneracy number of $\ket{GHZ}$ is 3, the degeneracy configuration of $\ket{GHZ}$ is ${\{1,1,1\}}.$  
 \section{Main Results}
\indent First we present the stabilizer group for a two-qubit symmetric pure state $\ket{\psi}\in\mathcal{H}_2\otimes \mathcal{H}_2$ is nontrivial. 
\begin{Theorem}
	Assume that $\ket{\psi}$ is a $2$-qubit symmetric pure state, $\ket{\psi}=\sum_{i=0}^2c_i\ket{D(2,i)},$ then the stabilizer group for the state $\ket{\psi}$ is nontrivial.
\end{Theorem}
\begin{proof}
Assume $A\in \mathcal{L}(\mathcal{H}_1)$ and $B\in\mathcal{L}(\mathcal{H}_2)$ and $A\otimes B\ket{\psi}=\ket{\psi},$ then we have $\mathcal{M}^{-1}(\ket{\psi})=\mathcal{M}^{-1}(A\otimes B\ket{\psi}),$ according to the equality $(8),$  we have that
\begin{align}
AX=X(B^{T})^{-1}.
\end{align}	
here we assume that
$
X=\left(
\begin{array}{cc}
	c_0 & c_1 \\
	c_1 & c_2 \\
\end{array}
\right),$ $A=\left(
\begin{array}{cc}
	 a&  b  \\
	 c& d \\
\end{array}
\right),B^{T}=\left(
\begin{array}{cc}
	h & -f \\
	-g &  e\\
\end{array}
\right)
$ and  $det(B)=1,$ then we can obtain the following equation set:
\begin{align}
\left\{
\begin{aligned}
c_0(c-a)=c_1(b-g)\\
c_0f-c_2b=c_1(a-h) \\
c_0c-c_2f=c_1(d-e)\\
c_1(f-c)=c_2(d-h)\\
\end{aligned}
\right  .
\end{align}
As there are four equations and eight variables, then we know that the rank of the solution vectors is more than 1, that is, the stabilizer group for the state $\ket{\psi}$ contains more than the identity.
\end{proof}
\indent Now we present a lemma to show an $n$-qubit symmetric pure state owns a nontrivial stabilizer group when $n=3,4$.
\begin{Lemma}
	Assume $\ket{\psi_{1}}$ and $\ket{\psi_2}$ are symmetric pure states, and $\ket{\psi_1}$ is SLOCC equivalent to $\ket{\psi_2},$ if there exists a nontrivial ILO $g$ such that $g^{\otimes n}\ket{\psi_{1}}=\ket{\psi_1},$ then the stabilizer group for $\ket{\psi_2}$ is nontrivial. 
\end{Lemma}  
\begin{proof}
	As $\ket{\psi_{1}}$ and $\ket{\psi_2}$ are symmetric states, then there exists an ILO $h$ such that $\ket{\psi_{2}}=h^{\otimes n}\ket{\psi_1}.$ $\ket{\psi_{2}}=h^{\otimes n}g^{\otimes n} (h^{-1})^{\otimes n}\ket{\psi_{2}},$ that is, the stabilizer group for $\ket{\psi_{2}}$ is nontrivial. 
\end{proof}
\indent In \cite{TSPM}, the authors presents a three-qubit symmetric pure state is SLOCC equivalent to $\ket{W}$ or $\ket{GHZ}.$ As we know, when we choose $g=\left(\begin{array}{cc}
exp(i\pi/4) &0\\0
& exp(-i\pi/2)\end{array}\right), g^{\otimes 3}\ket{W}=\ket{W}, \sigma^{\otimes 3}_z\ket{GHZ}=\ket{GHZ}.$ From Lemma 2, we see that the stabilizer group for all three-qubit pure symmetric states is nontrivial. And a four-qubit symmetric state is SLOCC equivalent to one of the elements in $S=\{\ket{D(4,0)},\ket{D(4,1)},\ket{D(4,2)},\frac{\ket{D(4,2)}+\ket{D(4,0)}}{2},\frac{\ket{GHZ}+\mu\ket{D(4,2)}}{\sqrt{1+\mu^2}}\},$ then we present a nontrivial stabilizer for the elements in the set $S,$ for the state $\ket{D(4,0)},$ we have $\sigma_z^{\otimes 4}\ket{D(4,0)}=\ket{D(4,0)},$ for the state $\ket{D(4,1)},$ we choose an ILO $g=\left(\begin{array}{cc}
\exp(\frac{\pi i}{6})&0\\
0&\exp(\frac{3\pi i}{2})
\end{array}\right),$ for the state $\ket{D(4,2)},$ we can choose an ILO $g=\sigma_z,$ and for the last element in the set S, we can also choose an ILO $g=\sigma_z,$ then due to the lemma 2, we have the stabilizer group for all four-qubit symmetric pure state is nontrivial. Note that for four-qubit pure states, the authors in \cite{BBM} also proposed the similar results. \\
 \indent Here we denote 
 \begin{align}
 G\equiv SL(2,\mathbb{C})\otimes SL(2,\mathbb{C})\otimes \cdots\otimes SL(2,\mathbb{C})\nonumber\\
 K\equiv SU(2)\otimes SU(2)\otimes \cdots\otimes SU(2)\nonumber\\
 \tilde{G}\equiv GL(2,\mathbb{C})\otimes GL(2,\mathbb{C})\otimes \cdots\otimes GL(2,\mathbb{C})\nonumber\\
 \tilde{K}\equiv U(2)\otimes U(2)\otimes \cdots\otimes U(2)
 \end{align}  Next we will use the method proposed in \cite{GKW} to give a class of symmetric states with its stabilizer group containing only the identity. First we introduce the definition of $SL$-invariant polynomials. A polynomial $f:\mathcal{H}\longrightarrow \mathbb{C}$ is $SL$-invariant if $f(g\ket{\psi})=f(\ket{\psi}),\forall g\in {SL}(2,\mathbb{C})\otimes {SL}(2,\mathbb{C})\cdots\otimes SL(2,\mathbb{C}),$ here we denote $\mathcal{H}=\mathcal{H}_2\otimes \mathcal{H}_2\otimes \cdots\otimes \mathcal{H}_2$ and $SL(2,\mathbb{C})=\{A=\left(\begin{array}{cc}
 a_{11}& a_{12}\\
 a_{21}& a_{22}
 \end{array}\right)| a_{ij}\in \mathcal{C},i,j=1,2,det(A)=1\}$. $f_2$ is a $SL$-invariant polynomial with degree 2, which is defined as $f_2(\ket{\psi})=(\ket{\psi},\ket{\psi})=\bra{\psi^{*}}\sigma_y^{\otimes n}\ket{\psi},$ here $\sigma_y$ is the Pauli operator with its matrix representation $\left(\begin{array}{cc}0&-i\\
 i&0\\
 \end{array}\right)$. Due to the property of $\sigma_y,$ we have that when $n$ is odd, $f_2(\cdot)=0.$ Another polynomial, $f_4(\ket{\psi}),$ is a polynomial with degree $4$, it is defined as
 $f_4(\ket{\psi})=\det\left(
 \begin{array}{cc}
 (\ket{\psi_0},\ket{\psi_0}) & (\ket{\psi_0},\ket{\psi_1}) \\
 (\ket{\psi_1},\ket{\psi_0}) &  (\ket{\psi_1},\ket{\psi_1})\\
 \end{array}
 \right),$ here we assume that $\ket{\psi}=\ket{0}\ket{\psi_0}+\ket{1}\ket{\psi_1}.$ Below we denote the stabilizer group for a pure state $\ket{\psi}$ as $\tilde{G}_{\psi}=\{g_1\otimes g_2\otimes \cdots\otimes g_n\in\tilde{G}|g_1\otimes g_2\otimes \cdots\otimes g_n\ket{\psi}=\ket{\psi}\}.$
 \begin{Example}
 	Assume $\ket{\psi}$ is an $n$-partite symmetric pure state $(n\ge 5)$, $\ket{\psi}=x_k\ket{D(n,k)}+x_l\ket{D(n,l)}+x_n\ket{D(n,n)},k+l=n+1,k,l\ne \lfloor\frac{n}{2}\rfloor, gcd(k,l)=1,$ here we denote that $gcd(k,l)$ is the greatest common divisor of $k$ and $l,$ Then $\tilde{G}_{\psi}=\{I\}$
 \end{Example}
 \begin{proof}
 	 First we denote $ g_2\otimes g_3\otimes \cdots\otimes g_n$ as $h$ and let $g_1=\left(\begin{array}{cc}
 	 a_1& b_1\\
 	 c_1& d_1\\\end{array}\right),$ then
 	\begin{align}
 	&\hspace{3mm} g_1\otimes h\ket{\psi}\nonumber\\
 	&=\ket{0}h[a_1(x_k\ket{D(n-1,k)}+x_l\ket{D(n-1,l)})+b_1(x_k\ket{D(n-1,k-1)}+x_l\ket{D(n-1,l-1)}\nonumber\\
 	&+x_n\ket{D(n-1,n-1)})]+\ket{1}h[c_1(x_k\ket{D(n-1,k)}+x_l\ket{D(n-1,l)})\nonumber\\
 	&+d_1(x_k\ket{D(n-1,k-1)}+x_l\ket{D(n-1,l-1)}+x_n\ket{D(n-1,n-1)})]\nonumber\\
 	&=\ket{0}(x_k\ket{D(n-1,k)}+x_l\ket{D(n-1,l)})\nonumber\\&+\ket{1}(x_k\ket{D(n-1,k-1)}+x_l\ket{D(n-1,l-1)}+x_n\ket{D(n-1,n-1)})
 	\end{align}
  \indent Apply $\bra{0}\otimes I^{\otimes (n-1)}$ to the left hand side (LHS) and the right hand side (RHS) of the above equality, and we denote that $\ket{\zeta_1}=x_k\ket{D(n-1,k)}+x_l\ket{D(n-1,l)},\ket{\zeta_2}=x_k\ket{D(n-1,k-1)}+x_l\ket{D(n-1,l-1)}+x_n\ket{D(n-1,n-1)},$ then
 	\begin{align}
&h[a_1\ket{\zeta_1}+b_1\ket{\zeta_2}]=x_k\ket{D(n-1,k)}+x_l\ket{D(n-1,l)}.
 	\end{align}
 \indent Assume $n$ is odd, then by using $f_2(\cdot)$ to the equality $(15),$ we have $2b_1^2x_kx_l=0,b_1=0.$ As $\ket{\psi}$ is a symmetric state, we may assume $g_i=\left(
 	\begin{array}{cc}
 	a_i & 0 \\
 	c_i &  d_i\\
 	\end{array}
 	\right),i=2,\cdots,n,$ due to $g_1\otimes g_2\otimes \cdots\otimes g_n\ket{\psi}=\ket{\psi} $ and through simple computation, we have that $c_i=0,$ Applying $g_i=\left(
 	\begin{array}{cc}
 	a_i & 0 \\ 	
 	0&  d_i\\
 	\end{array}
 	\right)$  to the equality $(14),$ we have $d_i=\lambda a_i,\prod d_i=1,\lambda^k=\lambda^l=1,$ as $gcd(k,l)=1,$ we have $\lambda=1,$ that is $\tilde{G}_{\psi}=\{I\}.$\\
 	\indent When $n$ is even, we use $f_4(\cdot)$ to the equality $(15)$, we have $2b_1^4x_k^2x_l^2=0,b_1=0.$ From the same method as when $n$ is odd, we have that $\tilde{G}_{\psi}=\{I\}$
  \end{proof}
\indent Then we present a class of symmetric critical states $\ket{\psi}$ with $G_{\psi}=\{I\}$. First we present the definition of critical states and a meaningful characterization of critical states. The set of critical states is defined as:
\begin{align}
Crit(\mathcal{H}_n)=\{\psi|\bra{\psi}X\ket{\psi}=0, X\in Lie(G)\}.
\end{align}
Here $Lie(G)$ is the Lie algebra of $G.$ The critical set is valuable, as many important states in quantum information theory, such as the Bell states, GHZ states, cluster states  and graph states, belong to the set of critical states. Then we present a fundamental characterization of critical states.
\begin{Lemma} (The Kempf-Ness theorem)\\
	1. $\psi\in Crit(\mathcal{H}_n)$ if and only if $||g\psi||\ge ||\psi||$ for all $g\in G,$ $||\cdot||$ denotes the Euclidean norm of $\ket{\psi}\in \mathcal{H}_n.$\\
	2. If $\psi\in Crit(\mathcal{H}_n),$ then $||g\psi||\ge ||\psi||$ with equality if and only if $g\psi\in K\psi.$ Moreover, if g is positive semidefinite then the equality condition holds if and only if $g\psi=\psi.$\\
	3. If $\psi\in\mathcal{H}_n,$ then $G\psi$ is closed in $\mathcal{H}_n$ if and only if $G\psi\cap Crit(\mathcal{H}_n)\ne \emptyset.$
\end{Lemma}  
\indent Due to this lemma, Gour and Wallach in \cite{GNW} showed that $\psi\in \mathcal{H}_n$ is critical if and only if all the local density matrices of $\ket{\psi}$ are proportional to the identity. And by the lemma in \cite{GKW} and the theorem 2.12 in \cite{W}, for a class of pure states $\ket{\psi}\in Crit(\mathcal{H}_n),$ if we could show the set $K_{\psi}=\{U_1\otimes U_2\otimes \cdots\otimes U_n\in K|U_1\otimes U_2\otimes \cdots\otimes U_n\ket{\psi}=\ket{\psi}\}$ contains only the identity, then $G_{\psi}=\{I\}.$
\begin{Example}
	Assume $\ket{\psi}= a_0\ket{D(n,0)}+a_{n-1}\ket{D(n,n-1)}$ or $\ket{\psi}=a_1\ket{D(n,1)}+a_n\ket{D(n,n)}$ is an $n$-qubit symmetric critical pure state with $n>4,$ then $G_{\psi}=\{I\}.$  
\end{Example} 
\begin{proof}
First we prove $K_{\psi}=\{I\}.$ 
 When $\ket{\psi}=a_1\ket{D(n,1)}+a_n\ket{D(n,n)},$ assume $U^{\otimes n}\in K$ satisfies $U^{\otimes n}\ket{\psi}=\ket{\psi},$ then 
	\begin{align}
	(U\otimes U)\rho_{1,2}(U\otimes U)^{+}=\rho_{1,2},
	\end{align} the equality $(17)$ can be changed as 
	\begin{align}
&\hspace{3mm}	(U\otimes U)[\frac{(n-2)|a_1|^2}{n}\ket{00}\bra{00}+\frac{|a_1|^2}{n}(\ket{01}+\ket{10})(\bra{01}+\bra{10})+|a_n|^2\ket{11}\bra{11}](U\otimes U)^{+}\nonumber\\=&\frac{|a_1|^2(n-2)}{n}\ket{00}\bra{00}+\frac{|a_1|^2}{{n}}(\ket{01}+\ket{10})(\bra{01}+\bra{10})+|a_n|^2\ket{11}\bra{11}
	\end{align}
	As the formula on the RHS of the equality above can be seen as a Schmidt decomposition, and $U\otimes U$ cannot increase the Schmidt rank, then $U\ket{1}=u_1\ket{1},$ as $U$ is a unitary operator, then $U\ket{0}=u_0\ket{0}.$ At last, due to $U^{\otimes n}\ket{\psi}=\ket{\psi},$ then we have $ u_0^{n-1}u_1=1,u_1^n=1,$ $i.e.$ $u_0=u_1.$ The other case is similar. 
\end{proof}
\indent  Here we present another proof on $G_{L_n}=\{I\},$ $\ket{L_n}$ is defined in the equality $(8)$ of the article \cite{GKW}.  Note that the examples above tells us $G_{\psi}=\{I\},$ however, $\tilde{G}_{\psi}\supset
\{I\}.$ It seems that this result is simple, However, this method is very useful to present nontrivial examples of states in n-qudit systems with nontrivial stabilizer groups \cite{SW}. \\
\indent At last, I would like to apply $M\ddot{o}bius$ transformation to show when the diversity number $m$ of an $n$-qubit symmetric pure state $\ket{\psi}$ is 5 or 6, the stabilizer group of $\ket{\psi}$ is trivial, when $m\ge 7,$ under a conjecture we make, the stabilizer group of $\ket{\psi}$ is trivial.\\
\indent Assume a pure symmetric state $\ket{\psi}$ can be represented in terms of Majorana representation:
\begin{align}
\ket{\psi}=\frac{e^{i\alpha}}{\sqrt{K}}\sum_{\sigma\in perm}\ket{\phi_{\sigma(1)}}\ket{\phi_{\sigma(2)}}\cdots\ket{\phi_{\sigma(n)}},
\end{align}
where the sum takes over all the permutations and K is the normalization for the state $\ket{\psi}.$ Due to the main results proposed by Mathonet $et$ $al.$ \cite{PSM}, we see that if $\ket{\psi}=g_1\otimes g_2\otimes\cdots \otimes g_n\ket{\psi},$ then there exists an ILO $g$ such that $\ket{\psi}=g^{\otimes n}\ket{\psi},$ $i.e.$ if we could prove $\{g|g^{\otimes n}\ket{\psi}=\ket{\psi}\}=\{I\},$ then the stabilizer group of $\ket{\psi}$ is trivial. From the Eq.$(19),$
\begin{align}
g^{\otimes n}\ket{\psi}=&\ket{\psi},\\
 \sum_{\sigma\in perm}\bigotimes_i g\ket{\phi_{\sigma(i)}}=&\sum_{\sigma\in perm}\bigotimes_i\ket{\phi_{\sigma(i)}}.
\end{align}
 Due to the uniqueness of the Majorana representation for a symmetric state and according to the equality (21), we see that there is a permutation $\sigma$ such that
\begin{align}
g\ket{\phi_i}=\lambda_i\ket{\phi_{\sigma(i)}}
\end{align}
with $\Pi_i \lambda_i=1.$ That is,
\begin{Theorem}
	Assume $\ket{\psi}$ is an $n$-qubit symmetric pure state, then the stabilizer group of $\ket{\psi}$ is nontrivial if and only if there exists nontrivial $M\ddot{o}bius$ transformation that permutes the Majorana roots of $\ket{\psi}.$
\end{Theorem}
\indent Assume  the diversity number of a symmetric state $\ket{\psi}$ is $m$, the divergence configuration for the state $\ket{\psi}$ is $\{k_1,k_2,\cdots,k_m\}$ with $k_1\ge k_2\cdots\ge k_m$ and $\sum_i k_i=n.$ Then from the Lemma $1$ and the property $(4)$ of $M\ddot{o}bius$ transformation, we have:
\begin{Corollary}
Assume $\ket{\psi}$ is an $n$-qubit symmetric state, its degeneracy configuration is $\{k_1,k_2,\cdots,k_m\},m\ge 3,$ if there exists $k_i,k_j$ and $k_l$ such that these three values are unequal to the residual elements in the set, then the stabilizer group for the state $\ket{\psi}$ contains only the identity.
\end{Corollary}
\indent At last, we talk about the stabilizer group for an $n$-qubit symmetric pure state $\ket{\psi}$ with the increase of the diversity number $m$ of $\ket{\psi},$ when $m=2$ the stabilizer group for $\ket{\psi}$ is nontrivial, when $m\ge 5,$ under a conjecture we make, the stabilizer group for $\ket{\psi}$ is trivial. \\
 $\textit{m=1}$ (separable states): When a symmetric state $\ket{\psi}$ is separable, then it can be represented as
 \begin{align}
 \ket{\psi}=\ket{\epsilon\epsilon\cdots\epsilon},
 \end{align}
 it is easy to see when an ILO $g$ satisfies that $\ket{\epsilon}$ is an eignvector of $g,$ $g^{\otimes n}$ is the stabilizer for the state $\ket{\psi}.$\\
 $\textit{m=2}:$ In this case, the state $\ket{\psi}$ can be represented as
 \begin{align}
 \ket{\psi}=\frac{e^{i\theta}}{\sqrt{K}}\sum \ket{\underbrace{\epsilon_1\cdots \epsilon_1}_{k_1}\underbrace{\epsilon_2\cdots\epsilon_2}_{n-k_1}},
 \end{align}
 where the sum takes over all the permutation of $k_1$ $\ket{\epsilon_1}s$ and $n-k_1$ $\ket{\epsilon_2}s$ in this case, when $k_1\ne n-k_1,$ first we denote a local operator $h=[\ket{\epsilon_1},\ket{\epsilon_2}]^{-1},$ this means if $\ket{\epsilon_1}=a_1\ket{0}+b_1\ket{1}, \ket{\epsilon_2}=a_2\ket{0}+b_2\ket{1},$ $h=\left(\begin{array}{cc}
 a_1&a_2\\
 b_1& b_2\\
 \end{array}\right)^{-1},$ then $g$ can be written as $g=h^{-1}\left(\begin{array}{cc}\lambda_1&0\\
  0&\lambda_2\end{array}\right) h.$ And when $k_1=n-k_1,$ the ILO $g$ can also be written as $g=h^{-1}\left(\begin{array}{cc}0&\lambda_2\\
  \lambda_1&0\end{array}\right)h, h=[\ket{\epsilon_1},\ket{\epsilon_2}]^{-1}$.\\
  \indent When $m\ge 3,$ due to Lemma 3, by searching a nontrivial M$\ddot{o}$bius transformation between the Majorana points of the pure state $\ket{\psi},$ we can see whether a pure state $\ket{\psi}$ owns a nontrivial stabilizer group.\\ \indent When $m=3,$ assume the degeneracy configuration of a pure state $\ket{\psi}$ is $\{k_1,k_2,k_3\},$ due to the property (5), we see that except when $k_1\ne k_2\ne k_3$, $\ket{\psi}$ owns a nontrivial stabilizer group.\\
  \indent Here we assume that each $z_i$ is not $\infty,$ as if there exists $k$ such that $z_k=\infty,$ we can always make $z_k$ be not $\infty$ by M$\ddot{o}$bius transformation. When $m=4,$ assume the degeneracy configuration of the pure state $\ket{\psi}$ is $\{k_1,k_2,k_3,k_4\},$ when all the four number are different from each other, the stabilizer group for $\ket{\psi}$ is trivial. First we show when $k_1=k_2,k_3=k_4,$ the stabilizer group is nontrivial, $i.e.$ there exists a nontrivial M$\ddot{o}$bius transformation $f$ that can permute $z_i,i=1,2,3,4.$ Let $f(z_1)=z_2,f(z_2)=z_1,f(z_3)=z_4,f(z_4)=z_3,$ from the property (6) of the M$\ddot{o}$bius transformation, we see $f$ exists.\\
  \indent Here we note that when $\ket{\psi}$ is a four qubit symmetric pure state, $\ket{\psi}$ always owns a nontrivial stabilizer group. As the diversity configuration can only be $\{1,1,1,1\},\{2,1,1\},\{3,1\}$ or $\{4\},$ from the above analysis, we see that the stabilizer group is nontrivial.\\
  \indent Next we consider a general case. Assume the stabilizer group of the state $\ket{\psi}$ is nontrivial, then we have a nontrivial M$\ddot{o}$bius transformation which permutes the majorana points, that is, we could assume $f(z_1)=z_2,f(z_2)=z_3,\cdots,f(z_{i_1})=z_1,f(z_{i_1+1})=z_{i_1+2},f(z_{i_1+2})=z_{i_1+3},\cdots, f(z_{i_2})=z_{i_1+1}$, if $i_1,i_2-i_1\ge 3,(i_1,i_2-i_1)=1,$ here we denote that $(i_1,i_2-i_1)$ is the greatest common divisor of $i_1$ and $i_2-i_1$, then we could have $f^{i_1}(z_j)=z_j,j=1,2,\cdots ,i_1$, that is, $f^{i_1}$ is trivial, when $i_1,i_2-i_1\ge 3,(i_1,i_2-i_1)=r\ne \min(i_1,i_2-i_1)$, $f^i(z_{j})\ne z_j, j\ge i_1+1$, so it is invalid.  However, we cannot prove the case when $(i_1,i_2-i_1)=min(i_1,i_2-i_1)$ is invalid. When $i_1=2$, $i_2-i_1\ge 3,$ then we have $f(z_1)=z_2,f(z_2)=z_1,f(z_3)=z_4,f(z_4)=z_5,$
  \begin{align}
  \frac{(z_1-z_2)(z_3-z_4)}{(z_1-z_4)(z_3-z_2)}=\frac{(z_2-z_1)(z_4-z_5)}{(z_2-z_5)(z_4-z_1)},\nonumber,
  \end{align}
that is,
\begin{align}
(z_3-z_4)(z_2-z_5)=(z_3-z_2)(z_4-z_5),\nonumber\\
(z_3-z_5)(z_2-z_4)=0.\nonumber
\end{align}
This case is invalid.\\
\indent  Assume $\ket{\psi}$ is an $n$-qubit symmetric pure state with its degeneracy configuration $\{k_1,k_2,\cdots,k_m\}$, and we divide the degeneracy configuration into $d$ parts according to whether $k_i$ are equal. Assume $D$ is a set with its elements $n_j,j=1,2,\cdots, d$, here $n_j$, $j=1,2,\cdots,d$ is the number of each part. Due to the above analysis, only when the $M\ddot{o}bius$ transformation $f$ is represented as $f_{\{{r,1}\}}(r\ge 1,r\in\mathbb{Z}^{+}),f_{\{\underbrace{r,r,\cdots,r}_l\}}(r\ge 2,l\ge 1,r,l\in \mathbb{Z}^{+})$, $f_{\{\underbrace{r,r,\cdots,r}_l,1\}}(r\ge 2,l\ge 2,r,l\in \mathbb{Z}^{+})$ or $f_{\{\underbrace{r,r,\cdots,r}_l,1,1\}} (r\ge 2,l\ge 1,r,l\in \mathbb{Z}^{+})$ may present $\ket{\psi}$ own a nontrivial stabilizer group.  That is, if $n_j$, $j=1,2,\cdots,d$ cannot divide into the form of subscript of $f$, then it owns a trivial stabilizer group. For example, if the degeneracy configuration of $\ket{\psi}$ is $\{5,5,5,5,5,3,3,3,3,3,3,3,2,2,2\}$, $D=\{5,7,3\}$, this set cannot be divide into the form of the above, that is, $\ket{\psi}$ owns a trivial stabilizer group.\\
\indent As the transformations of three points can determine a unique $M\ddot{o}bius$ transformation, and when we have a Majorana point in each part, we could determine whether the above configuration is valid, if all the permutations are invalid, then the symmetric state owns a trivial stabilizer group. So when the degeneracy number $m\ge 5$, the symmetric pure state $\ket{\psi}$ with trivial stabilizer group is of full measure among the subspace of symmetric states.\\  
\indent This fact is interesting, as it implies much in the entanglement theory. As when an n-qubit symmetric pure state is entangled, it is fully entangled. Due to the results in \cite{GKW}, we have the following theorem.
  \begin{Theorem}
   Assume $\ket{\psi}$ is an n-qubit symmetric pure state with $\tilde{G}_{\psi}=\{I\}$, $\ket{\phi}$ is an n-qubit symmetric entangled state, then $\ket{\psi}$ can be converted deterministically to $\ket{\phi}$ by LOCC or SEP operations if and only if there is $u\in \tilde{K}$ such that $\ket{\psi}=u\ket{\phi}.$ 
  \end{Theorem} 
\indent From the above theorem, we have if two states satisfy the conditions in the above theorem, if they are not LU equivalent, they cannot be transformed into other by LOCC, even SEP. And from the statement below the theorem 2, we have that when $m\ge 5$, almost all symmetric pure states $\ket{\psi}$ are isolated. This fact may represent the complexity of the structure of the multipartite entanglement.\\
\indent From the above results, we see that when two symmetric pure states satisfy conditions in the above theorems, a state can be transformed into the other by local transformations only with probability. At last, we present the maximal probability with which $\ket{\psi}\rightarrow \ket{\phi}$ can be converted by LOCC by using the theroem 7 in \cite{GKW}.
\begin{Theorem}
	Assume $\ket{\psi}$ is a symmetric pure state with $\tilde{G}_{\ket{\psi}}=\{I\}$, $\ket{\phi}=g^{\otimes n}\ket{\psi}$ is a symmetric pure state, then the maximal probability with which $\ket{\psi}$ can be converted to $\ket{\phi}$ by LOCC or SEP is given by
	\begin{align}
		p_{max}(\ket{\psi}\rightarrow\ket{\phi})=\frac{1}{\lambda_{max}({g^{\dagger}}^{\otimes n}g^{\otimes n})},
	\end{align}
	here we denote that $\lambda_{max}(X)$ is the maximal eigenvalue of the matrix $X$.
\end{Theorem}  
  
   \section{Conclusion}
\indent In this article, we consider the stabilizer group for a symmetric state $\ket{\psi}.$ First we present that the stabilizer group for an $n$-qubit symmetric state $\ket{\psi}$ contains more than the identity when $n=2,3,4$, then similar to the method presented in \cite{GKW}, we give a class of states whose stabilizer group contains only the identity, we also propose a class of states $\ket{\psi}$ with $G_{\psi}=\{I\},\tilde{G}_{\psi}\ne \{I\},$ at last, we present that when $m\ge 5$, almost all $n$-qubit symmetric pure states owns a trivial stabilizer group, then due to the results in \cite{GKW}, we have that almost all $n$-qubit symmetric pure states cannot be transformed among symmetric pure entangled states by nontrivial LOCC transformations deterministically, and we also present the optimal probability of the local transformations among two SLOCC-equivalent symmetric states.  
\section{ACKNOWLEDGMENTS}
	\indent This project was partially supported by the National Key Research and Development Program of China under grant 2016YFB1000902, the Natural Science Foundation of China grants 61232015 and 61621003, the Knowledge Innovation Program of the Chinese Academy of Sciences (CAS), and Institute of Computing Technology of CAS.

\end{document}